\documentclass[preprint,12pt,authoryear]{elsarticle}

\usepackage{amssymb}
\usepackage{amsthm}
\usepackage{amsmath}
\usepackage{color}

\newtheorem{Theorem}{Theorem}
\newtheorem*{ThmOld}{Theorem}

\newtheorem{Proposition}{Proposition}
\newtheorem{Corollary}{Corollary}
\newtheorem{Lemma}{Lemma}

\usepackage{comment}
\newcommand{\beq}{\begin{equation}}
\newcommand{\eeq}{\end{equation}}
\usepackage{hyperref}

\journal{}

\begin{document}

\begin{frontmatter}



\title{On the dimension of the solution space of linear difference equations over the ring of infinite sequences\\ \emph{\normalsize In memory of Marko Petkov\v sek}}


\author[1]{Sergei Abramov}
\ead{sergeyabramov@mail.ru}

\affiliation[1]{organization={Federal Research Center ``Computer Science and Control'' of the Russian Academy of Sciences},
            city={Moscow},
            country={Russia}}

\author[2]{Gleb Pogudin}
\ead{gleb.pogudin@polytechnique.edu}

\affiliation[2]{organization={LIX, CNRS, Ecole polytechnique, Institute Polytechnique de Paris},
    city={Paris},
    country={France}}

\begin{abstract}
For a linear difference equation with the coefficients being computable sequences, we establish algorithmic undecidability of the problem of determining the dimension of the solution space including the case when some additional prior information on the dimension is available.
\end{abstract}


\end{frontmatter}


\section{Introduction}

The present paper studies equations whose coefficients and solutions of interest are two-sided infinite sequences.
Infinite sequences are used in many areas of mathematics.
When working with these sequences
the way they are represented plays an important role. In this article, an algorithmic
approach is used: the sequence
is defined by an algorithm (each sequence has its own) for calculating the value of an element by the index
of this element.
More formally, we will call a two-sided sequence of rational numbers $\{v(n)\}_{n \in \mathbb{Z}}$ \emph{computable} if it is given by an algorithm computing the value of $v(n)$ for any given $n \in \mathbb{Z}$.
Other approaches are also possible, for example, if coefficient sequences satisfy linear
recurrences 
with constant coefficients, we naturally come to the concept of $C^2$-finite sequences
considered, for example, by~\cite{C2finite}.

Throughout the paper, $R$ stands for the ring of two-sided sequences having rational number terms with respect to termwise addition and multiplication.
For a linear difference equation
\beq
\label{eq}
a_r(n)y(n+r)+\dots +a_1(n)y(n+1)+a_0(n)y(n)=0
\eeq
with computable
$a_r(n),\dots,a_0(n)\in R$ as coefficients, we consider the $\mathbb{Q}$-vector space of solutions in~$R$.
In the sequel,
``equation'' is always understood as an equation of the form \eqref{eq}, and ``solution'' is a solution belonging to $R$.

In this paper, we establish algorithmic undecidability of the problem of
computing
the dimension of the solution space of an equation of the form~\eqref{eq} and the problem of
testing whether an equation has nonzero solutions at all (Section~\ref{sec:exist_nonzero}).
Furthermore, it is proven
that even when a finite set of possible values of the dimension of the solution space is known in advance,
if this set contains more than one element, then in the general case the dimension cannot be found algorithmically (Section~\ref{sec:general}).
On the other hand, we note the existence of naturally arising problems admitting an algorithmic solution (Section~\ref{solv}).
Our proofs are in general based on a consequence
of classical A.Turing's result on undecidability of the well-known halting problem \citep{Tur36}
which can be stated as follows.
\begin{ThmOld}[\cite{Tur36}]
Let $M$ be a set with $1 < |M| < \infty$.
Then there is no algorithm which, for a given computable  $c(0), c(1), \dots$ taking values in $M$ and any $a\in M$, determines whether the sequence contains an element equal to $a$.
\end{ThmOld}

This article is a continuation of the thematic line of publications \citep{ABP, acacca, nonlinear1, Pet2006, nonlinear2, nonlinear3}
on difference equations with general infinite sequences as coefficients and/or solutions.
Among these publications, the article by \cite{Pet2006} belongs to those works that provided a general basis for computer-algebraic studies of infinite sequences.

An abridged preliminary version of this paper appeared as a proceedings paper~\citep{Renat}.

\section{Possible value of dimension}

Before passing to the main topic of the paper, algorithmic questions related to computing the dimension of the solution space of~\eqref{eq}, we will show that in the case when the coefficients are computable sequences, there is no a priori relation between the order of the equations and the dimension of the solution space (in contrast to, say, the constant coefficient case).

\begin{Proposition}\label{lem:possible_dimensions}
  For every $r \in \mathbb{Z}_{\geqslant 0}$ and $d \in (\mathbb{Z}_{\geqslant 0} \cup \{\infty\})$, there exists an equation of the form~\eqref{eq} of order $r$ with both $a_0(n)$ and $a_r(n)$ not identical to zero which has $d$-dimensional solution space.
\end{Proposition}

Before proceeding to the proof of the proposition, we introduce one useful construction:
for difference equations~\eqref{eq}, we define \emph{interlacing} as follows.
Consider two such equations and assume that they both have order $r$,
since this can be achieved by adding several zero coefficients.
Denote their coefficients by $a_0(n), \ldots, a_r(n)$ and $b_0(n), \ldots, b_r(n)$, respectively.
We define sequences $c_0(n), \ldots, c_r(n)$ as follows
\[
c_i(n) := \begin{cases}
     a_i(n / 2), \text{ if $n$ is even},\\
     b_i((n - 1) / 2), \text{ if $n$ is odd}.
\end{cases}
\]
Now we define the interlacing of the original equations (we will denote it by the direct sum sign~$\oplus$) as the following equation of order $2r$:
\begin{equation}\label{eq:direct_sum}
c_r(n) y(n + 2r) + c_{r - 1}(n) y(n + 2r - 2) + \ldots + c_1(n) y(n + 2) + c_0(n) y(n) = 0.
\end{equation}
By construction, the solutions of the equations~\eqref{eq:direct_sum} are exactly sequences of $y(n)$ for which $y(2n)$ is a solution to the first equation and $y(2n + 1)$ is a solution to the second.

In particular, the dimension of the solution space of the constructed equation is the sum of the dimensions of the solution spaces of the original equations.
Interlacing of more than two terms is defined analogously.

\begin{proof}[Proof of Proposition~\ref{lem:possible_dimensions}]
  For $d \in (\mathbb{Z}_{\geqslant 0} \cup \{\infty\})$, we define a sequence $w_d(n)$ such that $w_d(n) = 0$ for $0 \leqslant n < d$ and $w_d(n) = 1$ otherwise.
  We denote by $E_d$ the equation $w_d(n) y(n) = 0$.
  We note that the dimension of the solution space of $E_d$ is equal to $d$.
  Equations $E_d$ prove the lemma for the case $r = 0$, so we will further focus on the case $r > 0$.

  For $r > 0$, we define an equation $E^\circ_r$ by $(1 - w_1(n)) y(n) + y(n + r) = 0$.
  For every $n \neq 0$, it implies that $y(n + r) = 0$, so $y(n) = 0$ for every $n \neq r$.
  Furthermore, by taking $n = 0$, we obtain $y(0) + y(r)$ which, together with $y(0) = 0$ established earlier, implies $y(r) = 0$.
  Thus, the only solution of $E^\circ_r$ is the zero solution.
  Now we consider an equation $E_d \oplus E^\circ_r$. It has order $r$ with the leading and trailing coefficients being not identically zero, and it has $d$-dimensional solution space.
\end{proof}


\section{Existence of nonzero solutions}
\label{sec:exist_nonzero}
\begin{Proposition}
\label{prp3}
Let
\beq
\label{seq}
v(0), v(1), \dots
\eeq
be a computable sequence.
Then a first-order difference equation (which will be called {\em signal}) can be presented, the dimension of the solution space of which is equal to $1$ if \eqref{seq} is identically zero, and is equal to $0$ (i.e., the equation has no nonzero solutions) otherwise.
\end{Proposition}

\begin{proof}
Based on the computable sequence $v(n)$, we define a computable two-sided sequence~$w(n)$:
\[
w(n) := \begin{cases}
     1, \text{ if $n<0$},\\
     1, \text{ if $n\geqslant 0$ and $v(k)=0$ for all $k=0, 1,\dots, n$},\\
     0, \text{ if $n\geqslant 0$ and $v(k) \neq 1$ for some $k$ such that $0\leqslant k \leqslant n$}.
\end{cases}
\]
By construction the sequence $w(n)$ consists of only ones if and only if the sequence $v(n)$ consists of only zeros.
If the sequence $v(n)$ contains at least nonzero element, then there is $n_0$ such that $w(n)=1$ for $n\leqslant n_0$ and $w(n)=0$ for $n> n_0$.
In the first case, the equation $y(n+1)-w(-n)y(n)=0$ has a solution space of dimension 1 (all constant sequences and only them will be the solutions),
in the second case, the solution space has dimension 0 (the equation has no nonzero solutions).
\end{proof}

The following theorem is a direct consequence of Proposition \ref{prp3}.
\begin{Theorem}
\label{prp3+}
(i) There is no algorithm that tests the existence of a nonzero solution to a given equation.

(ii) There is no algorithm that computes the dimension of the solutions space of a given equation.
\end{Theorem}

\begin{proof}
An algorithm that tests for the presence of nonzero solutions to a given equation would make it possible to check for the presence of nonzero elements in a given computable sequence $v(0), v(1),\dots$, which contradicts Turing's result \citep{Tur36}.
\end{proof}

\section{Computing dimension with a priori knowledge}
\label{sec:general}
In this section, we prove that even with some a priori restrictions on the dimension of the solution space, the problem of determining the exact dimension is still undecidable.

\begin{Theorem}\label{thm:general}
   For any subset $\mathcal{S} \subseteq (\mathbb{Z}_{\geqslant 0} \cup \{\infty\})$ with $|\mathcal{S}| > 1$,
   there is no algorithm that
   computes  the dimension of the solution space $d$ for a given equation of the form~\eqref{eq}, for which it is known in advance that $d\in \mathcal{S}$.
\end{Theorem}

\begin{proof}
   Consider two distinct elements $a$ and $b$ from $\mathcal{S}$.
   For an arbitrary computable sequence $v(0),\; v(1), \ldots$ we will construct an equation with the dimension of the solution space being $b$
   if $v(n)$ is identically zero and being $a$ otherwise.
   Then the undecidability of the problem of determining the dimension from the set $\{a, b\} \subseteq \mathcal{S}$ will
   follow from Turing's result~\citep{Tur36}.

   Consider the case $a, b \neq \infty$, and let $b > a$.
   Then we consider the equation $E_0$ having an $a$-dimensional space of solutions (for example, any equation of order $a$
   with constant nonzero coefficients) and the equation $E_1$ constructed from $v(n)$ in Proposition~\ref{prp3}.
   Consider the equation $E_0 \oplus \underbrace{E_1 \oplus \ldots \oplus E_1}_{b - a \text{ times}}$ (recall that $\oplus$
   stands for interlacing).
   The dimension of its solution space is equal to $b$ if $v(n)$ is identically zero and $a$ otherwise.
   Which is what we aimed at.

  Consider now the case when one of $a$ and $b$ is equal to infinity, let it be~$a$.
  We construct a sequence $w(n)$ such that $w(n) = 1$ for $n < 0$, $w(n) = 1$ for $n \geqslant 0$ if all $v(0), \ldots , v(n)$ are zeros and $w(n) = 0$ otherwise.
  We define the equation $E_2$ as $w(n) y(n) = 0$.
  If all elements of $v(n)$ are zeros, then $w(n) \equiv 1$, and hence the only solution is zero. 
  If a nonzero element occurs in $v(n)$, then there are infinitely many zeros in $w(n)$, and hence the dimension of the space of solutions will be infinite-dimensional.
  Let the equation $E_3$ be any equation that has an $b$-dimensional space of solutions.
  Then $E_2 \oplus E_3$ will be the desired equation.
\end{proof}

\begin{Corollary}\label{cor:determine_if_k}
For any non-negative integer $k$, there is no algorithm that would check if the solution space of the equation~\eqref{eq} has dimension~$k$.
\end{Corollary}

\begin{proof}
   If such an algorithm existed, then it could be used to calculate the dimension of the solution space in the case when it is known that the solution space is contained in the set $\mathcal{S} = \{k,\; k + 1\}$.
   This would contradict Theorem~\ref{thm:general}.
\end{proof}

\section{A case of decidability}
\label{solv}

Consider the case when the sequences $a_0(n), \ldots, a_r(n)$ are in fact periodic.
We will show that in this case, the dimension of the solution space can be computed using some standard tools from computer algebra, and the rest of the section will be devoted to proving the following proposition.
\begin{Proposition}\label{prop:constructive}
  There is an algorithm which takes as input periodic sequences $a_0(n), \ldots, a_r(n)$ and computes the dimension of the solution space of the equation
  \[
  a_r(n) y(n + r) + \ldots + a_1(n) y(n + 1) + a_0(n) y(n) = 0
  \]
  in the ring of two-sided sequences.
\end{Proposition}

Consider a positive integer $H > r$ such that the lengths of the periods of $a_0(n), \ldots, a_r(n)$ divide $H$ (such $H$ can always be taken to be a large enough common multiple of the period lengths).
We will ``decompose'' $y(n)$ into $H$ sequences $y_0(n) := y(H n), \; y_1(n) := y(Hn + 1), \;\ldots\;, y_{H - 1}(n) := y(Hn + H - 1)$.
Then the original equation~\eqref{eq} translates into the following $H$ linear difference equations with constant coefficients:
\begin{enumerate}
  \item $a_0(i) y_i(n) + \ldots + a_{r}(i) y_{i + r}(n) = 0$ for $0 \leqslant i < H - r$;
  \item $a_0(i) y_i(n) + \ldots + a_{H - 1 - i}(i) y_{H - 1}(n) + a_{H - i}(i) y_0(n + 1) + \ldots + a_r(i) y_{i + r - H}(n + 1) = 0$ for $H - r \leqslant i < H$.
\end{enumerate}
This way we have reduced the problem of computing the dimension of the solution space of~\eqref{eq} to the
problem of computing the dimension of the solution space of the system above.
We will state and solve this problem in a slightly more general form: for given $\ell \times H$ matrices $A_0$ and $A_1$, determine the dimension of the solution space of the system
\begin{equation}\label{eq:general_const}
A_0 \cdot (y_0(n), \;\ldots,\; y_{H - 1}(n))^T + A_1 \cdot (y_0(n + 1),\; \ldots,\; y_{H - 1}(n + 1))^T = 0.
\end{equation}
In order to do this, we consider a free module $F$ over the ring of Laurent polynomials $\mathbb{Q}[t, t^{-1}]$ with $H$ generators $e_0, \ldots, e_{H - 1}$.
Let $M$ be a submodule of $F$ generated by the entries of
\begin{equation}\label{eq:module_M}
  (A_0 + t A_1) \cdot (e_0,\;\ldots,\; e_{H - 1})^T.
\end{equation}

\begin{Lemma}\label{lem:dim}
  The dimension of the solution space of~\eqref{eq:general_const} is equal to the dimension of the quotient module $F / M$ over $\mathbb{Q}$.
\end{Lemma}

\begin{proof}
Let $S := F / M$.
We will define a linear bijective map between the solutions of~\eqref{eq:general_const} and linear functionals $S \to \mathbb{Q}$.
For a functional $\varphi\colon S \to \mathbb{Q}$, we define $y_i(n) = \varphi(t^n e_i)$.
Then the generators of $M$ and their translations by integer powers of $t$ will imply the equalities ~\eqref{eq:general_const}.

In the other direction, assume that we are given a solution of~\eqref{eq:general_const}.
We define $\widetilde{\varphi} \colon F \to \mathbb{Q}$ by $\widetilde{\varphi}(t^n e_i) = y_i(n)$.
Since the sequences $y_0(n), \ldots, y_{H - 1}(n)$ satisfy ~\eqref{eq:general_const}, we have that $\widetilde{\varphi}(M) = 0$.
Thus, $\widetilde{\varphi}$ induces a well-defined linear functional on the quotient module $F / M$.
\end{proof}

Thanks to Lemma~\ref{lem:dim}, the question of determining the dimension of the solution space of an equation~\eqref{eq} with periodic coefficients reduces to computing the dimension of the corresponding finitely presented module over the ring of Laurent polynomials.
The latter problem can be solved using Gr\"obner basis or using the Hermite normal form over the Laurent polynomial ring\footnote{We would like to thank Manuel Kauers for suggesting the approach via HNF.}.
In order to keep this note self-contained, we will present one way of computing this dimension based on a standard construction reducing the study of a module to the study of an ideal in a polynomial ring.

\begin{Lemma}\label{lem:module}
  In the notation of this section, we consider a polynomial ring $\mathcal{R} = \mathbb{Q}[t_1, t_{-1}, x_0, \ldots, x_{H - 1}]$ and an ideal
  \[
    I := \langle t_1t_{-1} - 1,\; (A_0 + t_1A_1) \cdot (x_0, \ldots, x_{H - 1})^T,\; \{x_i x_j | 0 \leqslant i, j < H\} \rangle.
  \]
  Then the $\mathbb{Q}$-dimension of the quotient module $F / M$ is equal to the $\mathbb{Q}$-dimension of the ideal generated by the images of $x_0, \ldots, x_{H - 1}$ in $\mathcal{R} / I$.
\end{Lemma}

\begin{proof}
  Let $J$ be the ideal in $\mathcal{R}$ generated by $I$ and $x_0, \ldots, x_{H - 1}$.
  We will denote its image in $\mathcal{R} / I$ by $\widetilde{J}$ which can be also viewed as an $\mathcal{R}/I$-module.
  Furthermore, since the multiplication by $x_0, \ldots, x_{H - 1}$ induces zero operator on $\widetilde{J}$, $\widetilde{J}$ is in fact a module over  $\mathcal{R} / J \cong \mathbb{Q}[t_{1}, t_{-1}] / \langle t_1t_{-1} - 1\rangle$.
  This ring is isomorphic to the ring of Laurent polynomials $\mathbb{Q}[t, t^{-1}]$ with the isomorphism given by $t_1 \to t, \; t_{-1} \to t^{-1}$.
  We define a homomorphism of modules over the Laurent polynomial ring $\varphi\colon F \to \widetilde{J}$ by $\varphi(e_i) = x_i$ for every $0 \leqslant i < H$.
  We observe that $\operatorname{Ker}(\varphi) \supset M$, so this homomorphism induces a surjective homomorphism $\widetilde{\varphi}\colon F / M \to \widetilde{J}$.
  Moreover, since $I$ contains all the degree two monomials in $x$'s, and linear relations on $x$'s are exactly~\eqref{eq:module_M}, $\widetilde{\varphi}$ is an isomorphism.
  The existence of such an isomorphism implies the equality of dimensions.
\end{proof}

Using Lemma~\ref{lem:module}, the dimension of $F / M$ can be determined by computing a Gr\"obner basis of $I$ and counting the monomials divisible by at least one of $x_0, \ldots, x_{H - 1}$ but not divisible by any of the leading monomials of the basis.
This completes the proof of Proposition~\ref{prop:constructive}.

Finally, we would like to point out that every periodic sequence satisfies a linear recurrence with constant coefficients (in other words, belongs to the class of $C$-finite sequences), so the class of sequences considered in this section is closely related with $C^2$-finite sequences studied in~\citep{C2finite}.
More precisely, $C^2$-finite sequence is a solution of~\eqref{eq} with the coefficients being $C$-finite such that the leading coefficient $a_r(n)$ does not contain zeros.
The latter condition considerably simplifies the problem of computing the dimension of the solution space (in particular, implies that this dimension is always finite).
A natural generalization of the problem studied in this section would be a problem of computing the dimension of the solution space of~\eqref{eq} with all coefficients being $C$-finite (that is, without requiring the absence of zeros in $a_r(n)$).
We do not know if this problem is algorithmically decidable.

\subsection*{Acknowledgements}

We would like to thank the referees for careful reading and detailed comments which helped us improve the manuscript.



\bibliographystyle{elsarticle-harv} 
\bibliography{bib}

\begin{thebibliography}{9}
\expandafter\ifx\csname natexlab\endcsname\relax\def\natexlab#1{#1}\fi
\providecommand{\url}[1]{\texttt{#1}}
\providecommand{\href}[2]{#2}
\providecommand{\path}[1]{#1}
\providecommand{\DOIprefix}{doi:}
\providecommand{\ArXivprefix}{arXiv:}
\providecommand{\URLprefix}{URL: }
\providecommand{\Pubmedprefix}{pmid:}
\providecommand{\doi}[1]{\href{http://dx.doi.org/#1}{\path{#1}}}
\providecommand{\Pubmed}[1]{\href{pmid:#1}{\path{#1}}}
\providecommand{\bibinfo}[2]{#2}
\ifx\xfnm\relax \def\xfnm[#1]{\unskip,\space#1}\fi
\bibitem[{Abramov et~al.(2021)Abramov, Barkatou and Petkov{\v{s}}ek}]{ABP}
\bibinfo{author}{Abramov, S.}, \bibinfo{author}{Barkatou, M.A.},
  \bibinfo{author}{Petkov{\v{s}}ek, M.}, \bibinfo{year}{2021}.
\newblock \bibinfo{title}{Linear difference operators with coefficients in the
  form of infinite sequences}.
\newblock \bibinfo{journal}{Computational Mathematics and Mathematical Physics}
  \bibinfo{volume}{61}, \bibinfo{pages}{1582--1589}.
\newblock \URLprefix \url{https://doi.org/10.1134/s0965542521100018}.
\bibitem[{Abramov and Pogudin(2023a)}]{acacca}
\bibinfo{author}{Abramov, S.}, \bibinfo{author}{Pogudin, G.},
  \bibinfo{year}{2023}a.
\newblock \bibinfo{title}{Linear difference operators with sequence
  coefficients having infinite-dimentional solution spaces}.
\newblock \bibinfo{journal}{{ACM} Communications in Computer Algebra}
  \bibinfo{volume}{57}, \bibinfo{pages}{1--4}.
\newblock \URLprefix \url{https://doi.org/10.1145/3610377.3610378}.
\bibitem[{Abramov and Pogudin(2023b)}]{Renat}
\bibinfo{author}{Abramov, S.}, \bibinfo{author}{Pogudin, G.},
  \bibinfo{year}{2023}b.
\newblock \bibinfo{title}{On the solution space of linear difference equations
  over the ring of computable sequences}, in: \bibinfo{booktitle}{Proceedings
  of the XIV conference "Differential equations and related topics" (June
  16-17, Kolomna, Moscow region)}, pp. \bibinfo{pages}{9--15}.
\newblock \bibinfo{note}{In Russian}.
\bibitem[{Jim{\'{e}}nez-Pastor et~al.(2023)Jim{\'{e}}nez-Pastor, Nuspl and
  Pillwein}]{C2finite}
\bibinfo{author}{Jim{\'{e}}nez-Pastor, A.}, \bibinfo{author}{Nuspl, P.},
  \bibinfo{author}{Pillwein, V.}, \bibinfo{year}{2023}.
\newblock \bibinfo{title}{An extension of holonomic sequences: {$C^2$}-finite
  sequences}.
\newblock \bibinfo{journal}{Journal of Symbolic Computation}
  \bibinfo{volume}{116}, \bibinfo{pages}{400--424}.
\newblock \URLprefix \url{https://doi.org/10.1016/j.jsc.2022.10.008}.
\bibitem[{Ovchinnikov et~al.(2020)Ovchinnikov, Pogudin and
  Scanlon}]{nonlinear1}
\bibinfo{author}{Ovchinnikov, A.}, \bibinfo{author}{Pogudin, G.},
  \bibinfo{author}{Scanlon, T.}, \bibinfo{year}{2020}.
\newblock \bibinfo{title}{Effective difference elimination and
  {N}ullstellensatz}.
\newblock \bibinfo{journal}{Journal of the European Mathematical Society}
  \bibinfo{volume}{22}, \bibinfo{pages}{2419--2452}.
\newblock \URLprefix \url{https://doi.org/10.4171/jems/968}.
\bibitem[{Petkov{\v{s}}ek(2006)}]{Pet2006}
\bibinfo{author}{Petkov{\v{s}}ek, M.}, \bibinfo{year}{2006}.
\newblock \bibinfo{title}{Symbolic computation with sequences}.
\newblock \bibinfo{journal}{Programming and Computer Software}
  \bibinfo{volume}{32}, \bibinfo{pages}{65--70}.
\newblock \URLprefix \url{https://doi.org/10.1134/s0361768806020022}.
\bibitem[{Pogudin et~al.(2020)Pogudin, Scanlon and Wibmer}]{nonlinear2}
\bibinfo{author}{Pogudin, G.}, \bibinfo{author}{Scanlon, T.},
  \bibinfo{author}{Wibmer, M.}, \bibinfo{year}{2020}.
\newblock \bibinfo{title}{Solving difference equations in sequences:
  {U}niversality and {U}ndecidability}.
\newblock \bibinfo{journal}{Forum of Mathematics, Sigma} \bibinfo{volume}{8}.
\newblock \URLprefix \url{https://doi.org/10.1017/fms.2020.14}.
\bibitem[{Turing(1936)}]{Tur36}
\bibinfo{author}{Turing, A.M.}, \bibinfo{year}{1936}.
\newblock \bibinfo{title}{On computable numbers, with an application to the
  {E}ntscheidungsproblem}.
\newblock \bibinfo{journal}{Proceedings of the London Mathematical Society}
  \bibinfo{volume}{2}, \bibinfo{pages}{230--265}.
\bibitem[{Wibmer(2021)}]{nonlinear3}
\bibinfo{author}{Wibmer, M.}, \bibinfo{year}{2021}.
\newblock \bibinfo{title}{On the dimension of systems of algebraic difference
  equations}.
\newblock \bibinfo{journal}{Advances in Applied Mathematics}
  \bibinfo{volume}{123}, \bibinfo{pages}{102136}.
\newblock \URLprefix \url{https://doi.org/10.1016/j.aam.2020.102136}.

\end{thebibliography}


\end{document}